\documentclass[letterpaper]{article} 
\usepackage{aaai25}  
\usepackage{aaai25}  
\usepackage{times}  
\usepackage{helvet}  
\usepackage{courier}  
\usepackage[hyphens]{url}  
\usepackage{graphicx} 
\usepackage{amsmath}
\usepackage{amssymb}
\usepackage{amsthm}
\usepackage{booktabs}
\usepackage{multirow}
\usepackage{natbib}  
\usepackage{caption} 
\usepackage{algorithm}
\usepackage{algorithmic}

\usepackage{newfloat}
\usepackage{listings}
\DeclareCaptionStyle{ruled}{labelfont=normalfont,labelsep=colon,strut=off} 
\floatstyle{ruled}
\newfloat{listing}{tb}{lst}{}
\floatname{listing}{Listing}
\title{Addressing Cold-Start Problem in Click-Through Rate Prediction via 
  Supervised Diffusion Modeling}
\author{
  Wenqiao Zhu,
  Lulu Wang,
  Jun Wu
}
\affiliations{
  HiThink Research \\

    \{zhuwenqiao, wanglulu2, wujun2\}@myhexin.com
}

\usepackage{bibentry}

\newtheorem{lemma}{Lemma}

\begin{document}

\maketitle

\begin{abstract}
  Predicting Click-Through Rates is a crucial function within recommendation and advertising platforms, as the output of CTR prediction determines the order of items shown to users. The Embedding \& MLP paradigm has become a standard approach for industrial recommendation systems and has been widely deployed. However, this paradigm suffers from cold-start problems, where there is either no or only limited user action data available, leading to poorly learned ID embeddings. The cold-start problem hampers the performance of new items. To address this problem, we designed a novel diffusion model to generate a warmed-up embedding for new items.
Specifically, we define a novel diffusion process between the ID embedding space and the side information space.
In addition, we can derive a sub-sequence from the diffusion steps to expedite training, given that our diffusion model is non-Markovian.
Our diffusion model is supervised by both the variational inference and binary cross-entropy objectives, enabling it to generate warmed-up embeddings for items in both the cold-start and warm-up phases.
Additionally, we have conducted extensive experiments on three recommendation datasets. The results confirmed the effectiveness of our approach.  
\end{abstract}

\begin{links}
     \link{Code}{https://github.com/WNQzhu/CSDM}
\end{links}

\section{Introduction}
Recommendation systems are crucial components of numerous commercial platforms, addressing the challenge of information overload prevalent in the digital age. A primary goal of many such systems often involves predicting Click-Through Rates (CTR).
Embedding \& MLP (Multilayer Perceptron) methods \cite{10.5555/3172077.3172127, 10.1145/3447548.3467077,10.1609/aaai.v33i01.33015941} have been widely utilized for this purpose. However, the Embedding \& MLP approach faces the cold-start problem, which arises from the long-tail distribution of candidate items and the dynamic nature of real-world recommendation systems.
New items in recommendation systems often have no or very limited user interactions.
 Consequently, the embeddings for these items are not adequately learned, leading to sub-optimal performance in predicting their CTRs.
 Furthermore, the embedding layer makes up a substantial part of the model's parameters and determines the input for the feature interaction and MLP modules.
Thus, it is imperative to improve the embedding layer to mitigate the cold-start problem in CTR prediction tasks.

To address the challenge of the cold-start problem, several approaches have been proposed that leverage either the limited available samples or the side information associated with new items. Approaches that make use of the limited available samples often employ a meta-learning framework, such as MAML \cite{Finn2017ModelAgnosticMF}, to derive robust representations for new items through an optimized training procedure \cite{10.1145/3331184.3331268,Lee2019MeLUMU,Lu2020MetalearningOH,ouyang2021learning,10.1145/3404835.3462843}. On the other hand, methods that leverage side information generally learn a transformation from user/item attributes to robust embeddings \cite{10.5555/2832747.2832769,10.1145/2959100.2959172,Saveski2014ItemCR,10.1145/564376.564421,10.1145/1995966.1995976,10.1145/3474085.3475665,10.1145/3397271.3401178}.

The aforementioned approaches learn representations as fixed points in the embedding space. However, due to the limited data available in cold-start scenarios, it is highly challenging to learn a reliable representation \cite{9010263}. To achieve reliable embedding learning for cold-start items, some variational approaches have been proposed \cite{10.1145/3485447.3512048,zhao2022improving}.
These approaches treat embedding learning as a distribution estimation problem and have shown effectiveness in tackling the cold-start problem. However, these approaches struggle with the trade-off between traceability and flexibility \cite{10.5555/3157382.3157627,10.5555/3045118.3045358,wang2023diffrec}, and they also suffer from the model collapse problem \cite{10.5555/3524938.3525758}.


To address the limitations of existing methods and further improve the performance of \textbf{c}old-\textbf{s}tart in CTR prediction, we propose a \textbf{d}iffusion \textbf{m}odel named CSDM, which constructs the transition between embeddings and side information in a denoising manner.
While diffusion models \cite{ho2020denoising,10.5555/3045118.3045358} have achieved remarkable results in image synthesis tasks \cite{teng2023relay} and can address the limitations in variational auto-encoder methods, it is not straightforward to adopt diffusion models to tackle the cold-start problem directly due to the following reasons:
(1) Standard diffusion models construct transitions between the data distribution and standard Gaussian distribution; in the context of the cold-start problem, we need to construct transitions between embeddings and side information.
(2) Diffusion models require more training and inference time due to the lengthy denoising steps.

We propose a novel diffusion model to overcome the aforementioned obstacles in addressing the cold-start problem.
Specifically, in addition to gradually adding noise during the forward process of diffusion steps, we also progressively incorporate portions of side information. This approach enables us to construct a transition between ID embeddings and side information. We consider our model as a non-Markovian model, which permits us to extract a sub-sequence during the generation process to expedite the training speed. We update the original embeddings with the newly generated ones, ensuring no additional inference cost is incurred during the inference phase.
We perform extensive experiments on three CTR prediction benchmark datasets to validate the effectiveness of our proposed method.

In a nutshell, the contributions of this work include:
\begin{itemize}
\item We propose a diffusion model to address the cold-start problem in CTR prediction tasks while considering both training and inference costs. To the best of our knowledge, we are the first to employ diffusion models for cold-start problems in CTR predictions.
\item We design a new diffusion process that allows us to construct transitions between ID embeddings and side information. Furthermore, our model is non-Markovian, which enables us to extract sub-sequences to reduce training costs.
\item Extensive experiments are conducted on three public benchmark datasets, and the results show that CSDM outperforms existing cold-start methods in CTR predictions.
\end{itemize}

\section{Preliminary: Diffusion Models}
Diffusion models represent a class of generative models that
leverage the diffusion process to remove noise from latent samples incrementally,
resulting in the generation of new samples.
DDPM \cite{ho2020denoising} is one of the
most representative diffusion models, comprising both a forward process and a reverse process.

\textbf{Forward process} gradually adds noises
to the given original data $\mathbf{z}_0 \in R^d$ over a sequence of  $T$ steps,
creating a Markov chain $\mathbf{z}_0, \mathbf{z}_1, \cdots, \mathbf{z}_T$.
In this chain, $\mathbf{z}_T$ is assumed to be an approximation of Gaussian noise.
Each forward step is defined as:
\[
q(\mathbf{z}_t|\mathbf{z}_{t-1}) := \mathcal{N}\left(\mathbf{z}_t;
\sqrt{\frac{\alpha_t}{\alpha_{t-1}}} \mathbf{z}_{t-1}, \left(1 - \frac{\alpha_t}{\alpha_{t-1}}\right)\mathbf{I}\right).
\]
DDPM incorporates a pre-established, constant noise schedule $\alpha_{1:T} \in (0,1]^T$, which controls the quantity of noise introduced at each step.
It admits a closed form of $\mathbf{z}_t$ at any timestep $t$: $q(\mathbf{z}_t|\mathbf{z}_0) = \mathcal{N}(\mathbf{z}_t; \sqrt{\alpha_t} \mathbf{z}_0, (1 - \alpha_t)\mathbf{I})$.

\textbf{Reverse process} sequentially eliminates the noise from $\mathbf{z}_t$ to 
recover $\mathbf{z}_{t-1}$ in a recursive manner, continuing this process until it reaches the initial step 0.
A single transition step parameterized by
$p_\omega(\mathbf{z}_{t-1}|\mathbf{z}_{t}) := \mathcal{N}(\mathbf{z}_{t-1}; \mathbf{\mu}_\omega(\mathbf{z}_t, t),
\mathbf{\Sigma}_\omega(\mathbf{z}_t, t))$ is learned,
where $\mathbf{\mu}_\omega(\mathbf{z}_t, t)$ and $\mathbf{\Sigma}_\omega(\mathbf{z}_t, t)$
are learned mean and variance. A U-net \cite{Ronneberger2015UNetCN} architecture is employed to model these parameters.

\textbf{Optimization.} Given the definition of the forward and reverse process,
DDPM optimizes the Evidence Lower Bound (ELBO) objective function.
It calculates the KL-divergence between $p_\omega$ and $q$ plus an entropy term:
\begin{align}
  \mathcal{L}_{\text{diff}} &= \mathbb{E}_q \left[\underbrace{D_{KL}(q(\mathbf{z}_t|\mathbf{z}_0) \Vert p(\mathbf{z}_T))}_{\mathcal{L}_T}\right] \notag\\
  & +
  \mathbb{E}_q \left[\sum_{t > 1} \underbrace{D_{KL} (q(\mathbf{z}_{t-1}|\mathbf{z}_t, \mathbf{z}_0) \Vert
    p_\omega(\mathbf{z}_{t-1} | \mathbf{z}_t))}_{\mathcal{L}_{t-1}} \right] \notag\\
  & - \mathbb{E}_q \left[ \underbrace{\log p_\omega(\mathbf{z}_0 | \mathbf{z}_1)}_{\mathcal{L}_0}  \right]  \label{eq:opt}
\end{align}
An inherent limitation of DDPM is its dependence on a Markovian process, which leads to a computationally intensive reverse process. To address this limitation, DDIM \cite{song2020denoising} adopts a non-Markovian method, markedly accelerating the reverse process. The objective function of DDIM is equivalent to Equation (\ref{eq:opt}), up to a constant difference.
\section{Method}
\begin{figure*}[t]
\begin{center}
  \includegraphics[width=\linewidth]{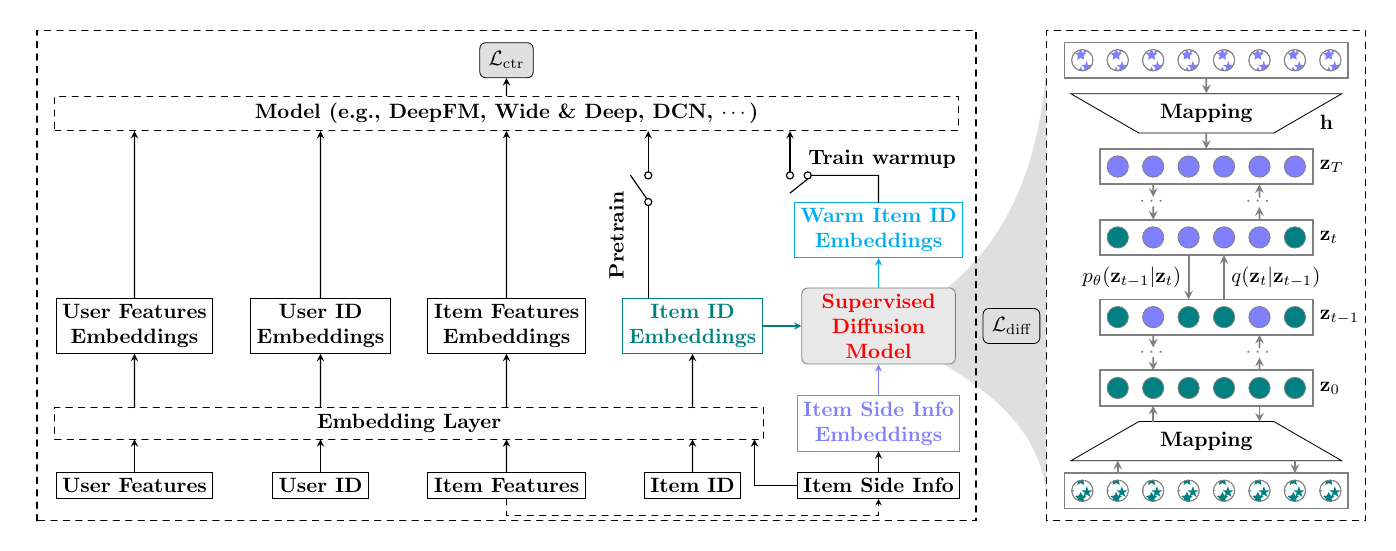}
\end{center}
\caption{The proposed CSDM framework for cold-start problems in CTR prediction.}
\label{fig:framework}
\end{figure*}
\subsection{Problem Definition}
The objective of Click-Through Rate (CTR) prediction is to forecast the likelihood
that a user will click on a specific presented item.
The outcome of this prediction will determine the final ordering of items shown to users.
A CTR prediction task is typically structured as a supervised
binary classification problem and is trained using an i.i.d (independent and identically distributed) dataset $\mathcal{D}$ from users' historical interactions.
Each instance $(\mathbf{x}, y) \in \mathcal{D}$ includes a collection of features $\mathbf{x}$ and a target label $y \in \{0, 1\}$, indicating the user's reaction to the presented item.
Generally, the input feature $\mathbf{x}$ can be decomposed into several components:
\begin{equation}
  \mathbf{x} = [u, \mathcal{X}_u, i, \mathcal{X}_i, \textit{context}]
\end{equation}
Here, $u$ is a unique identifier for each user within the recommendation system. Similarly, $i$ is a unique identifier for each item.
$\mathcal{X}_u$ represents the set of features associated with users, while $\mathcal{X}_i$ represents the set of features with items. The $\textit{context}$ refers to environmental features such as time and location.

The Embedding \& MLP paradigm first employs embedding technology
to convert the IDs and features into unique embeddings.
We designate the embeddings for the item ID, use ID, item features, user features, and context features as 
$\mathbf{e}_i \in R^d$, $\mathbf{e}_u \in R^d$, $\mathbf{e}_{\mathcal{X}_i} \in R^{d\times |\mathcal{X}_i|}$, $\mathbf{e}_{\mathcal{X}_u} \in R^{d\times |\mathcal{X}_u|}$, $\mathbf{e}_c \in R^{d\times |c|}$, respectively. Here, $d$ is the dimension of the embeddings.
The CTR model estimates the probability $\hat{y} = \textit{Pr}(y=1|\mathbf{x})$ by applying a
discriminative function $f(\cdot)$:
\begin{equation}
\hat{y} = f(\mathbf{e}_i, \mathbf{e}_u, \mathbf{e}_{\mathcal{X}_i}, \mathbf{e}_{\mathcal{X}_u}, \mathbf{e}_c;\theta, \phi)
\end{equation}
where $\theta$ denotes the parameters of the backbone deep model $f(\cdot)$ and
$\phi$ denotes the parameters of the embedding layers. The Binary Cross Entropy is often employed as the loss function for binary classification:
\begin{equation}
  \mathcal{L}_{ctr}(\theta,\phi) = -y \log \hat{y} - (1-y)\log(1-\hat{y})
  \label{eq:l_ctr}
\end{equation}
Since the parameters are trained by users' historical behavior data,
the embedding $\mathbf{e}_i$ for recently emerged items $(i, \mathcal{X}_i)$
with limited user interactions are not well learned.
As a result, the backbone model $f(\cdot)$ struggles to make an accurate 
estimation of the probability for these items.
This issue is known as the item cold-start problem, which includes two stages: (1) The cold-start phase, during which there are no user interactions with the item, and (2) the warm-up phase, where there are a limited number of user interactions. In this work, we tackle both of these stages, focusing exclusively on the item cold-start problem.
Specifically, a subset of item features $\mathcal{X}_v \subset \mathcal{X}_i$ are employed with our diffusion model to address the cold-start problem in CTR prediction.
\subsection{Supervised Diffusion Model}
To address the cold-start problem, we utilize a subset of item features $\mathcal{X}_v$, referred to as side information, to enrich the item ID embeddings. To achieve this, we employ a diffusion model to enable the flow of semantic information between the side information embeddings and the ID embeddings. However, since the standard diffusion process transforms an embedding into Gaussian noise, it is not directly applicable in this context. Therefore, we design a new diffusion process to learn warm-up ID embeddings for items, as illustrated in Figure \ref{fig:framework}.

We first pre-train a backbone model to provide the initial item ID embeddings. Then, we conduct the diffusion process to generate warmed-up embeddings.
Following \cite{10.5555/3600270.3600583,Cui2024DiffusionbasedCL}, we convert the discrete side information to a continuous space using an embedding map and project the initial item ID embeddings into hidden states. Let $\mathbf{h}$ and $\mathbf{z}_0$ denote the hidden state of the side information and the initial embedding, respectively. We gradually transform $\mathbf{z}_0$ into $\mathbf{h}$ using a forward diffusion process:
\begin{equation}
\mathbf{z}_0 \rightarrow \mathbf{z}_1 \rightarrow \cdots \rightarrow \mathbf{z}_T = \mathbf{h} + \epsilon, 
\end{equation}
where $\epsilon \sim \mathcal{N}(\mathbf{0}, \mathbf{I})$.
Since our diffusion process depends on both $\mathbf{z}_0$ and $\mathbf{h}$, it is no longer
Markovian. Inspired by DDIM \cite{song2020denoising}, we define a family $\mathcal{Q}$ of inference distributions, indexed by a real vector $\sigma \in \mathbb{R}^T_{\ge 0}$:
\begin{align}
  q_\sigma(\mathbf{z}_{1:T}|\mathbf{z}_0, \mathbf{h}) :=
  q_\sigma(\mathbf{z}_{T} | \mathbf{z}_0, \mathbf{h})
  \prod_{t=2}^{T} q_\sigma\left(\mathbf{z}_{t-1}|\mathbf{z}_t, \mathbf{h}, \mathbf{z}_0\right)
  \notag
\end{align}
where
\begin{equation}
  q_\sigma(\mathbf{z}_{T}| \mathbf{z}_0, \mathbf{h}) =
  \mathcal{N}\left(\sqrt{\alpha_{T}}\mathbf{z}_0 + \sqrt{c_{T}}\mathbf{h}, (1-\alpha_{T})\mathbf{I}\right)
  \label{eq:z_T}
\end{equation}
The forward process is governed by a decreasing sequence $\alpha_t \in (0, 1]^T$ and an increasing sequence $ c_t \in (0, 1]^T$. In ideal case, when $\alpha_t \rightarrow 0$ and
    $c_t \rightarrow 1$, we have $\mathbf{z}_T \sim \mathcal{N}(\mathbf{h}, \mathbf{I})$. 
In the recommendation scenario, since \(\mathbf{z}_t\) contains collaborative filtering information and \(\mathbf{h}\) contains feature information, it is not necessary to fully zero out \(\mathbf{z}_0\) or \(\mathbf{h}\).       
    Furthermore, we choose the  mean function as
\[
  q_\sigma(\mathbf{z}_{t-1}|\mathbf{z}_t,\mathbf{z}_0, \mathbf{h}) =
  \mathcal{N}(\mathbf{z}_{t-1}|\kappa_t \mathbf{z}_t + \lambda_t \mathbf{z}_0 + \nu_t \mathbf{h}, \sigma_t^2 \mathbf{I})
\]
in order to guarantee that the following equation 
\begin{equation}
  q_\sigma(\mathbf{z}_{t}| \mathbf{z}_0, \mathbf{h}) =
  \mathcal{N}(\sqrt{\alpha_{t}}\mathbf{z}_0 + \sqrt{c_{t}}\mathbf{h}, (1-\alpha_{t})\mathbf{I})
  \label{eq:z_t}
\end{equation}
holds true for all $t$.
The parameters are set as:
\begin{align}
  &\kappa_t = \sqrt{\frac{1-\alpha_{t-1} - \sigma^2}{1-\alpha_t}} \notag \\
  &\lambda_t = \sqrt{\alpha_{t-1}} - \sqrt{\alpha_t} \sqrt{\frac{1-\alpha_{t-1} - \sigma^2}{1-\alpha_t}} \notag \\
  &\nu_t = \sqrt{c_{t-1}} - \sqrt{c_t} \sqrt{\frac{1-\alpha_{t-1} - \sigma^2}{1-\alpha_t}} \notag
\end{align}
This posterior function can be derived from Bayes' rule, we show the proof in the supplementary materials.

In the reverse process, given $\mathbf{z}_t$, we can predict the denoised observation of the hidden state of ID embeddings:
\begin{equation}
  g_\omega^{(t)}(\mathbf{z}_t) := (\mathbf{z}_t - \sqrt{c_t}\mathbf{h}-\sqrt{1 - \alpha_t} \epsilon_\omega^t(\mathbf{z}_t))/\sqrt{\alpha_t}
\end{equation}
where $\{\epsilon_\omega^{(t)}\}_{t=1}^T$ is a set of $T$ functions to predict noise from $\mathbf{z}_t$ and $\omega$ contains the learnable parameters of diffusion model.
Therefore, we can define the reserve step as:
\begin{equation}
  p_\omega^{(t)}(\mathbf{z}_{t-1}|\mathbf{z}_t) =
  \left\{
  \begin{array}{lr}
    g_\omega^{(1)}(\mathbf{z}_t) & \text{if $t = 1$};  \\
    q_\sigma(\mathbf{z}_{t-1}|\mathbf{z}_t, g_\omega^{(t)}(\mathbf{z}_t), \mathbf{h})   & \text{otherwise}
  \end{array}
  \right.
  \notag
\end{equation}
\textbf{Optimization.} The parameters of diffusion models are optimized by a combination of the $\mathcal{L}_{\text{ctr}}$ and
the variational inference objective $\mathcal{L}_{\text{diff}}$:
\begin{equation}
  \mathcal{L} = \mathcal{L}_\text{ctr} + \rho \mathcal{L}_\text{diff},
  \label{eq:loss}
\end{equation}
where $\rho$ is a hyper-parameter.
By combining $\mathcal{L}_\text{ctr}$ and $\mathcal{L}_\text{diff}$, the generative process of the diffusion model takes into account both the collaborative filtering information derived from user action data and the side information obtained from item features. Consequently, our method is applicable in both the cold-start and warm-up stages of the CTR prediction task.
In our implementation, we utilize the simplified version of $\mathcal{L}_\text{diff}$, as proposed in DDPM \cite{ho2020denoising}.
\begin{table*}[!t]
  \begin{center}
    \begin{tabular}{cc||lr||lr||lr||lr}
      \toprule
      &\multirow{2}{*}{Methods} &  \multicolumn{2}{c||}{Cold} & \multicolumn{2}{c||}{Warm-a} & \multicolumn{2}{c||}{Warm-b} & \multicolumn{2}{c}{Warm-c} \\
      && AUC & RelaImpr & AUC & RelaImpr & AUC & RelaImpr & AUC & RelaImpr \\
      \cmidrule{2-10}
      \multirow{7}{*}{\rotatebox[origin=c]{90}{\centering MovieLens-1M}}&
      DeepFM        & 0.7313  & 0.00\%  & 0.7464 & 0.00\% & 0.7588 & 0.00\% & 0.7692 & 0.00\%\\  
      &DropoutNet(DeepFM)   & \underline{0.7410}  & \underline{4.19}\%       & 0.7506 & 1.70\%      & 0.7593 & 0.19\%      & 0.7671 &  -0.78\%\\
      &MWUF(DeepFM) & 0.7324  &  0.47\%      & 0.7466 & 0.08\%      & 0.7590 &  0.07\%     & 0.7694 & 0.07\%\\
    &Meta-E(DeepFM) &0.7397  &  3.63\%      & 0.7513 &  1.98\%     & 0.7614 & 1.00\%      & 0.7690 & -0.07\%\\
      &VELF(DeepFM) & 0.7244  & -2.98\%       & 0.7695 & 9.37\%      & 0.7382 & -7.95\%      & 0.7741 & 1.82\%\\
      &CVAR(DeepFM) & 0.7349  & 1.55\%       & \underline{0.7910} & 18.10\%      & \underline{0.8009} &  \underline{16.27}\%     & \underline{0.8044} & \underline{13.07}\%\\
      &CSDM(DeepFM) & \textbf{0.7443}  &  \textbf{5.62}\%      & \textbf{0.7982} & \textbf{21.02}\%
                    & \textbf{0.8058}  &  \textbf{18.16}\%      & \textbf{0.8089} & \textbf{14.74}\%\\
      \midrule
      &\multirow{2}{*}{Methods} &  \multicolumn{2}{c||}{Cold} & \multicolumn{2}{c||}{Warm-a} & \multicolumn{2}{c||}{Warm-b} & \multicolumn{2}{c}{Warm-c} \\
      && AUC & RelaImpr & AUC & RelaImpr & AUC & RelaImpr & AUC & RelaImpr \\
      \cmidrule{2-10}
      \multirow{7}{*}{\rotatebox[origin=c]{90}{\centering Taobao AD}}&
      DeepFM        & 0.5958  & 0.00\%  & 0.6089 & 0.00\%  & 0.6204 & 0.00\% & 0.6306  & 0.00\%\\  
      &DropoutNet(DeepFM)   & 0.5970  &  1.25\%      & 0.6097 & 0.73\%       & 0.6207  & 0.25\%      & 0.6305 & -0.7\%\\
      &MWUF(DeepFM) & 0.5967   &  0.94\%     & 0.6101 & 1.10\%      &  0.6207 & 0.25\%      &  0.6303& -0.23\%\\
      &Meta-E(DeepFM) & 0.5975   & 1.77\%       & 0.6119 & 2.75\%      & 0.6226  &  1.83\%     & 0.6323 & 1.30\%\\
      &VELF(DeepFM) & 0.5967   &  0.93\%      & 0.6176 & 7.98\%      & 0.6258  & 4.49\%      & 0.6335 & 2.22\%\\
      &CVAR(DeepFM) & \underline{0.5998}&  \underline{4.17\%}      & \underline{0.6194} &  \underline{9.64\%}     & \underline{0.6295}  & \underline{7.56\%}     & \underline{0.6370} & \underline{4.90\%}\\
      &CSDM(DeepFM) & \textbf{0.6004}   & \textbf{4.80\%}       & \textbf{0.6290} & \textbf{18.45\%}      & \textbf{0.6324}  & \textbf{9.97\%}      & \textbf{0.6382} & \textbf{5.82\%}\\
      \midrule
      &\multirow{2}{*}{Methods} &  \multicolumn{2}{c||}{Cold} & \multicolumn{2}{c||}{Warm-a} & \multicolumn{2}{c||}{Warm-b} & \multicolumn{2}{c}{Warm-c} \\
      && AUC & RelaImpr & AUC & RelaImpr & AUC & RelaImpr & AUC & RelaImpr \\
      \cmidrule{2-10}
      \multirow{7}{*}{\rotatebox[origin=c]{90}{\centering CIKM 2019}}&
      DeepFM        & 0.7376  & 0.00\%  & 0.7522 & 0.00\%  & 0.7605 & 0.00\% & 0.7671  & 0.00\%\\  
      &DropoutNet(DeepFM)   & 0.7367  & -0.38\%       & 0.7487 &  -1.39\%      & 0.7569  & -1.38\%      & 0.7636 & -1.31\%\\
      &MWUF(DeepFM) & 0.7372   & -0.17\% &0.7501  & -0.83\%      & 0.7598  & -0.27\%      & 0.7674 & 0.11\%\\
      &Meta-E(DeepFM) & 0.7367 & -0.38\%      &0.7483  & -1.55\%      & 0.7574  & -1.19\%      & 0.7651 & -0.75\%\\
      &VELF(DeepFM) & 0.7403   & 1.13\%      & 0.7393 &  -5.11\%     & 0.7390  &  -8.25\%     & 0.7317 & -13.25\%\\
      &CVAR(DeepFM) & \underline{0.7405}   & \underline{1.22\%}     & \underline{0.7588} & \underline{2.62\%}      & \underline{0.7649}  &  \underline{1.69\%}     & \underline{0.7687} & \underline{0.60\%}\\
      &CSDM(DeepFM) & \textbf{0.7418}  & \textbf{1.77\%}       & \textbf{0.7624} & \textbf{4.04\%}      & \textbf{0.7686}  &  \textbf{3.10\%}    & \textbf{0.7710} & \textbf{1.46\%} \\
      \bottomrule
    \end{tabular}
  \end{center}
  \caption{Model comparison on three datasets. DeepFM is utilized as the backbone. Ten runs are conducted for each method. The best and second-best improvements are highlighted in bold and underlined, respectively.}
  \label{tb:main-res}
\end{table*}
\begin{figure*}[!t]
\begin{center}
  \includegraphics[width=\linewidth]{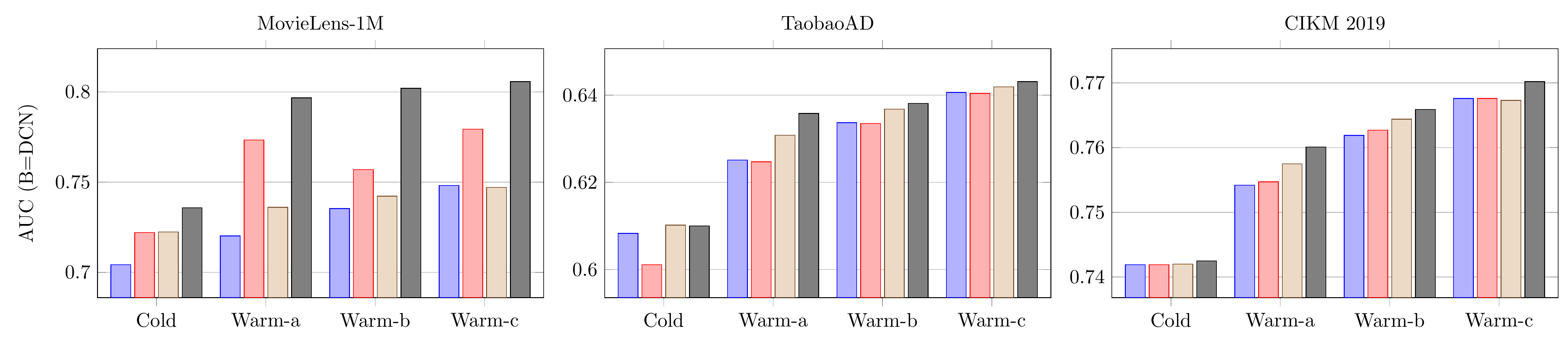}
\end{center}
\begin{center}
  \includegraphics[width=\linewidth]{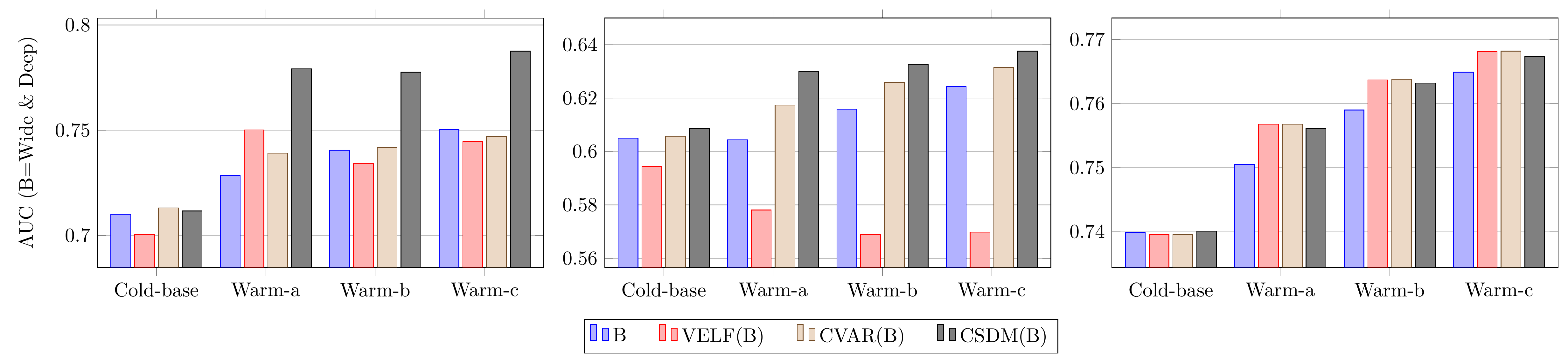}
\end{center}

\caption{AUC scores evaluated across various stages for different backbone models, conducted over three datasets with 10 runs per model.}
\label{fig:various-backbones}
\end{figure*}

\begin{figure}[!t]
\begin{center}
  \includegraphics[width=\linewidth]{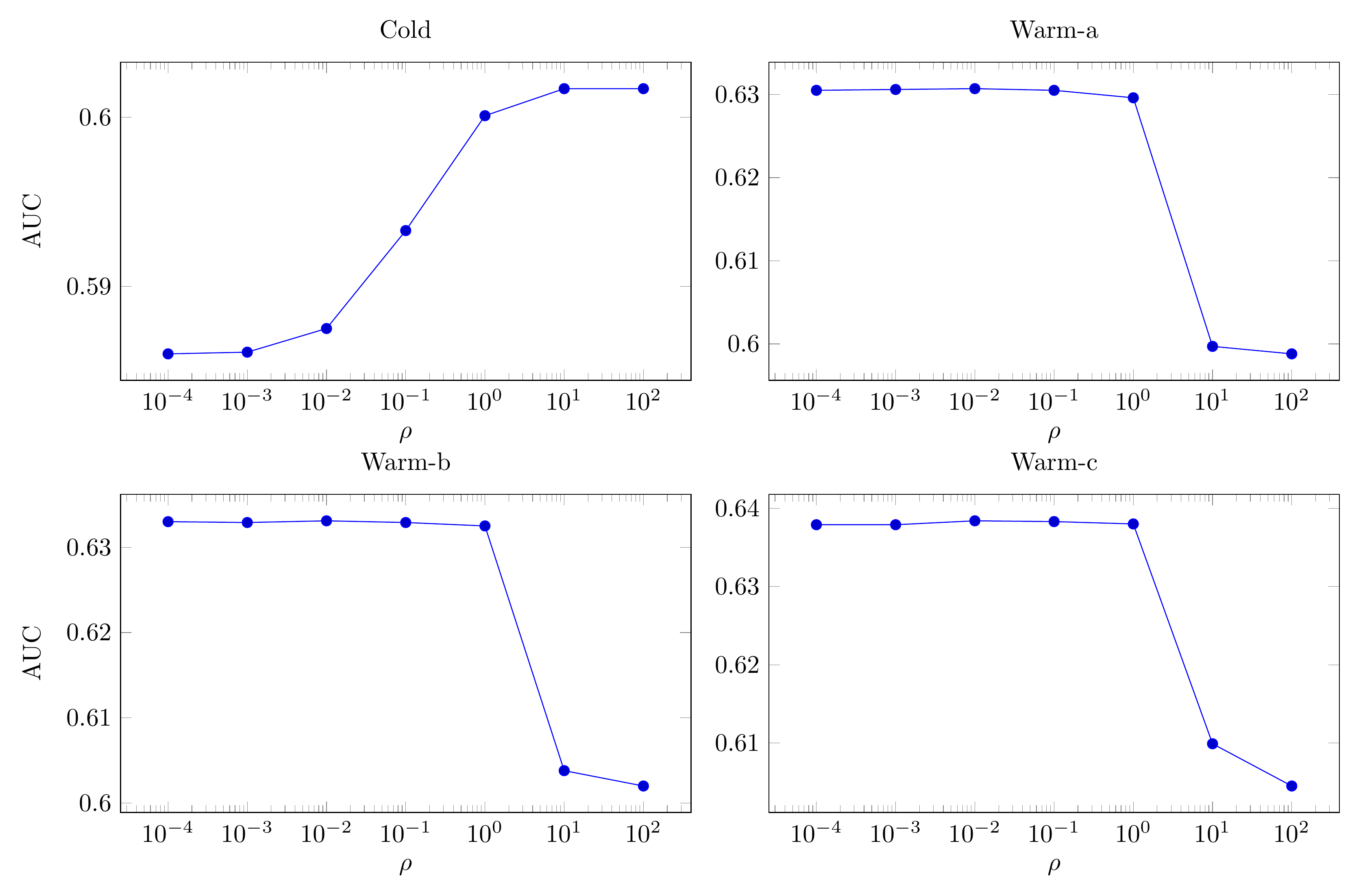}
\end{center}

\caption{Performance evaluation on the TaobaoAD dataset using DeepFM as the backbone model across a range of $\rho$ values, mean of three runs.}
\label{fig:lambda}
\end{figure}
\textbf{Generating warmed-up embeddings.}
Given the reverse step provided above, we can generate the warmed-up embeddings from $\mathbf{h}$ by repeating the following step:
\begin{align}
\mathbf{z}_{t-1} &=\sqrt{\alpha_{t-1}}\underbrace{\left(\frac{\mathbf{z}_t - \sqrt{c_t}\mathbf{h} - \sqrt{1-\alpha_t} \epsilon^{(t)}_\omega (\mathbf{z}_t)}{\sqrt{\alpha_t}}
  \right)}_{\text{predicted hidden state of initial ID embedding}} \notag \\
&+\underbrace{\sqrt{1 - \alpha_{t-1} - \sigma_t^2} \epsilon^{(t)}_\omega(\mathbf{z}_t)}_{\text{direction pointing to $\mathbf{z}_t$}} \notag \\
&+\sqrt{c_{t-1}}\underbrace{\mathbf{h}}_{\text{side information}} + \underbrace{\sigma_t \epsilon_t}_{\text{random noise}}
\label{eq:gen}
\end{align}
From Equation (\ref{eq:gen}), we can observe that the generating process consists of two key components: (1) predicting the relevant ID embedding, and (2) leveraging side information to improve the generation quality. In this process, we set $\sigma_t = 0$ for all $t$ to ensure that the reverse process is deterministic.
Upon obtaining the estimated warmed-up embedding $\tilde{\mathbf{z}}_0$, we project it back into the ID embedding space via a linear function. Following the approach taken by CVAR \cite{zhao2022improving}, we incorporate the frequency of items into this linear function. The original ID embeddings are subsequently replaced with the newly generated ones. As a result, our method does not modify the backbone structures and incurs no extra computational overhead during the inference phase. 
\\\textbf{Sub-sequence.}
Our approach is non-Markovian. Therefore, we can consider a sub-sequence of the latent variables $\mathbf{z}_{1:T}$ to accelerate the generative process. We uniformly sample a sub-sequence from $\mathbf{z}_{1:T}$ using a step parameter $s$. A larger $s$ results in a smaller sub-sequence, which in turn enables a faster generative process.

\section{Experiment}
\subsection{Dataset}
We evaluate our method on three publicly available datasets: \textbf{MovieLens-1M
  \footnote{http://www.grouplens.org/datasets/movielens/}}, \textbf{Taobao Display Ad Click
  \footnote{https://tianchi.aliyun.com/dataset/dataDetail?dataId=56}}, and \textbf{CIKM 2019 EComm AI
  \footnote{https://tianchi.aliyun.com/competition/entrance/231721}}.
The details of these datasets are described in the supplementary material.
\subsection{Baselines}
We compare our method with two groups of Click-Through Rate (CTR) prediction methods. The first group encompasses a variety of common feature-crossing techniques tailored specifically for CTR prediction. Methods within this group also serve as the foundational model for numerous cold-start algorithms.
(1) DeepFM \cite{10.5555/3172077.3172127} is a method that combines low-order feature interactions through Factorization Machines (FM) \cite{Rendle2010FactorizationM} and high-order feature interactions through a deep neural network.
(2) Wide \& Deep \cite{10.1145/2988450.2988454} contains a linear model and a deep neural network to effectively handle both simple and complex relationships in the data.
(3) DCN \cite{10.1145/3124749.3124754} explicitly models the interactions between features using a deep network.

The other group comprises state-of-the-art methods aimed at addressing the cold-start problem in CTR prediction tasks. We compare our method with the following methods: DropoutNet \cite{10.5555/3295222.3295249}, MWUF \cite{10.1145/3404835.3462843}, Meta-E \cite{10.1145/3331184.3331268}, VELF \cite{10.1145/3485447.3512048}, and CVAR \cite{zhao2022improving}.

\subsection{Experimental Settings}
\textbf{Dataset splits.} We divided the datasets into several groups following \cite{10.1145/3404835.3462843} to assess the performance of our proposed method in both the cold-start and warm-up phases.
First, we divide the items into two groups based on their frequency using a threshold $N$.
Items with a frequency greater than $N$ are classified as old items, while those with a lower frequency are considered new items. The threshold $N$ is set to 200 for MovieLens-1M, 2000 for TaobaoAD, and 200 for CIKM 2019, ensuring the ratio of new items to old items is approximately 8:2, which mirrors a long-tail distribution as described in \cite{10.1145/3397271.3401043}.
We further divide the new item instances, sorted by timestamps, into four groups: warm-a, warm-b, warm-c, and a test set. The first $3\times K$ instances are distributed evenly among warm-a, warm-b, and warm-c, with the remainder allocated to the test set. The value of $K$ is set to 20 for MovieLens-1M, 500 for TaobaoAD, and 50 for CIKM 2019, respectively.
\\\textbf{Implementation Details.}
We apply consistent experimental settings across all methods for each dataset to ensure fair comparisons. The embedding size for all features is set to 16 for compared methods.
Additionally, the MLPs in the backbone models utilize two dense layers, each with 16 units.
We set the learning rate to 0.001 for all methods, and the mini-batch size is set to 2048 for MovieLens-1M and TaobaoAD, and 4096 for CIKM 2019. All methods are optimized using the Adam optimizer \cite{2015-kingma} on shuffled samples. We set \( T=100 \) for the total number of forward steps.
The parameters \(\rho\) and \( s \) are searched from the sets \(\{0.001, 0.01, 0.1, 1\}\) and \(\{5, 10\}\), respectively.
We define the sequences \(\{\alpha_t\}\) and \(\{c_t\}\) using a hyper-parameter \(\beta=10^{-5}\) for all experiments:
\begin{equation}
  \alpha_t = (1-\beta)^t
\end{equation}
\begin{equation}
  c_t = \left(\sum_{k=1}^t\sqrt{\frac{\alpha_t}{\alpha_k}}\right) / 
  \left(\sum_{k=1}^T\sqrt{\frac{\alpha_t}{\alpha_k}}\right)
\end{equation}
Clearly, \(\{\alpha_t\}\) is a decreasing sequence and \(\{c_t\}\) is an increasing sequence.
We implement the U-Net as a two-layer MLP module with position encoding \cite{10.5555/3295222.3295349} enabled.
Furthermore, To avoid over-fitting in the diffusion model, we adopt dropout with $p=0.5$ in our implementation in the sample process and U-Net. Specifically, we have 
\begin{equation}
  \mathbf{z}_t = \sqrt{\alpha_t} \mathbf{z}_0 + \sqrt{c_t} \text{drop}(\mathbf{h}, p=0.5)
  + \sqrt{1-\alpha_t}\epsilon
\end{equation}
in the forward process, where $\epsilon \in \mathcal{N}(\mathbf{0},\mathbf{I})$.
It's worth noting that the dropout utilized here differs from DropoutNet in that Dropout affects the backbone model, whereas ours impacts the diffusion model only.
\\\textbf{Evaluation metrics.}
We use the Area Under the Curve (AUC) \cite{Ling2003AUCAS} as the metric to evaluate performance. This is a widely used metric in both recommendation systems \cite{10.1145/3383313.3412236} and computational advertising \cite{10.1145/3219819.3219823,10.1145/3331184.3331268}. An AUC value of 0.5 corresponds to random guessing. \cite{10.5555/3044805.3044982} proposed a relative improvement (RelaImpr) metric to assess the performance improvement, which is calculated as follows:
\begin{equation}
\text{RelaImpr} = \left(\frac{\text{AUC(model)} - 0.5}{\text{AUC(baseline)} - 0.5} - 1\right) \times 100\%
\end{equation}
We use this metric to compare the relative improvement in performance across all methods.

\subsection{Experiment Results}
\textbf{Comparison with State-of-the-arts.}
Our approach refines the embeddings of cold item IDs through a diffusion process, all without altering the underlying model architecture. Consequently, we compare our CSDM method against a range of state-of-the-art cold-start solutions from an embedding learning perspective, encompassing DropoutNet, MWUF, Meta-E, VELF, and CVAR. Simultaneously, we select DeepFM, a well-known method for CTR prediction, as the base model. We assess the average outcomes across ten runs on three distinct datasets and present these results in Table \ref{tb:main-res}.

We can notice that CSDM outperforms other comparative baselines across all datasets. The results indicate that our supervised diffusion model is capable of learning high-quality initial embeddings. Apart from learning better initial embeddings, our approach also demonstrates enhanced utilization of user action data during the warm-up phase. This is corroborated by the noticeable performance gains observed in the warm-up stage, as shown in Table \ref{tb:main-res}.

In addition to the aforementioned observations, we have noticed a decline in the relative improvement as items reach maturity. This trend can be attributed to the growing influence of user action data as items garner more interactions. The progressive increase in AUC scores from Warm-a to Warm-b, and Warm-c stages underscores this trend. In later stages, the backbone model is also capable of learning more refined embeddings, given that these items have been interacted with by a larger user base.
\\\textbf{Generalization Experiments.}
Our approach refines the ID embeddings to tackle the cold-start challenge in CTR prediction, making it model-agnostic.
We demonstrate its versatility by performing experiments on a range of backbone models beyond DeepFM, including Wide \& Deep and DCN.
More experiments are provided in the supplementary material.
For each model variant, we execute 10 trials to obtain the average AUC across three datasets, with the results reported in Figure \ref{fig:various-backbones}.

Upon analyzing the results, we note the following: (1) Our method generally outperforms the baseline model, demonstrating its effectiveness. (2) Furthermore, our method typically exceeds VELF and CVAR in most cases, highlighting the advantage of the diffusion method over the variational approach for cold-start CTR prediction tasks.
\\\textbf{Ablation Study.}
\begin{table}[t]
  \setlength{\tabcolsep}{3pt}  
  \begin{center}
    \begin{tabular}{c|c| c|c|c|c}
      \toprule
      Dropout & Dataset & Cold & Warm-a & Warm-b & Warm-c \\
     \midrule      
     w  & ML-1M  & 0.7456 & 0.7980 & 0.8059 & \textbf{0.8091} \\
     w/o  & ML-1M & 0.7456 & 0.7980 & 0.8059 & 0.8090 \\
     \midrule
     w  & TaobaoAD & \textbf{0.6002} & 0.6296 & 0.6324 & \textbf{0.6383} \\
     w/o& TaobaoAD & 0.5859 & \textbf{0.6306} & \textbf{0.6329} & 0.6379 \\
     \midrule
     w  & CIKM     & \textbf{0.7418} & \textbf{0.7622} & \textbf{0.7687}& \textbf{0.7711} \\
     w/o& CIKM     & 0.7402 & 0.7587 & 0.7677 & 0.7705 \\
    \bottomrule
    \end{tabular}
  \end{center}
  \caption{An ablation test on the dropout function of diffusion models: "w" indicates that dropout is enabled, whereas "w/o" signifies that dropout is disabled. ML-1M stands for MovieLens-1M.}
  \label{tb:drop-deepfm}
\end{table}
We conduct ablation tests on our CSDM to determine the impact of its parameters on performance. First, we evaluate the performance of CSDM across different warm-up phases with varying $\rho$. The experiments are conducted over the TaobaoAD dataset using DeepFM as the backbone model. The results are presented in Figure \ref{fig:lambda}.

It is observable that the CSDM's performance in relation to $\rho$ varies between the cold phase and the warm-up phase. In the cold phase, an increase in $\rho$ generally leads to improved performance. Conversely, in the warm-up stage, an excessively large $\rho$
results in a notable decrease in performance.
The reason is that in the cold phase, the cold items have only the side information and
$\mathcal{L}_{ctr}$ only contributes to the hot items. Increasing $\rho$ can emphasize the contribution of the diffusion model in this setting.
Additionally, this finding demonstrates the effectiveness of the CSDM approach in the cold-start problem. Meanwhile, in the warm-up stage, new items also have some user interaction. Enlarging $\rho$ excessively can harm the contribution of user action data.

We conduct an ablation analysis on the diffusion model concerning the dropout mechanism across three datasets. The results are shown in Table \ref{tb:drop-deepfm}.
It can be observed that the effect of dropout varies across the datasets. On the MovieLens-1M dataset, the inclusion of dropout provides minimal improvement. In contrast, on the CIKM 2019 dataset, incorporating dropout into the diffusion process positively impacts and enhances the model's performance.
\\\textbf{Computational Overhead.}
\begin{figure}[t]
\begin{center}
  \includegraphics[width=\linewidth]{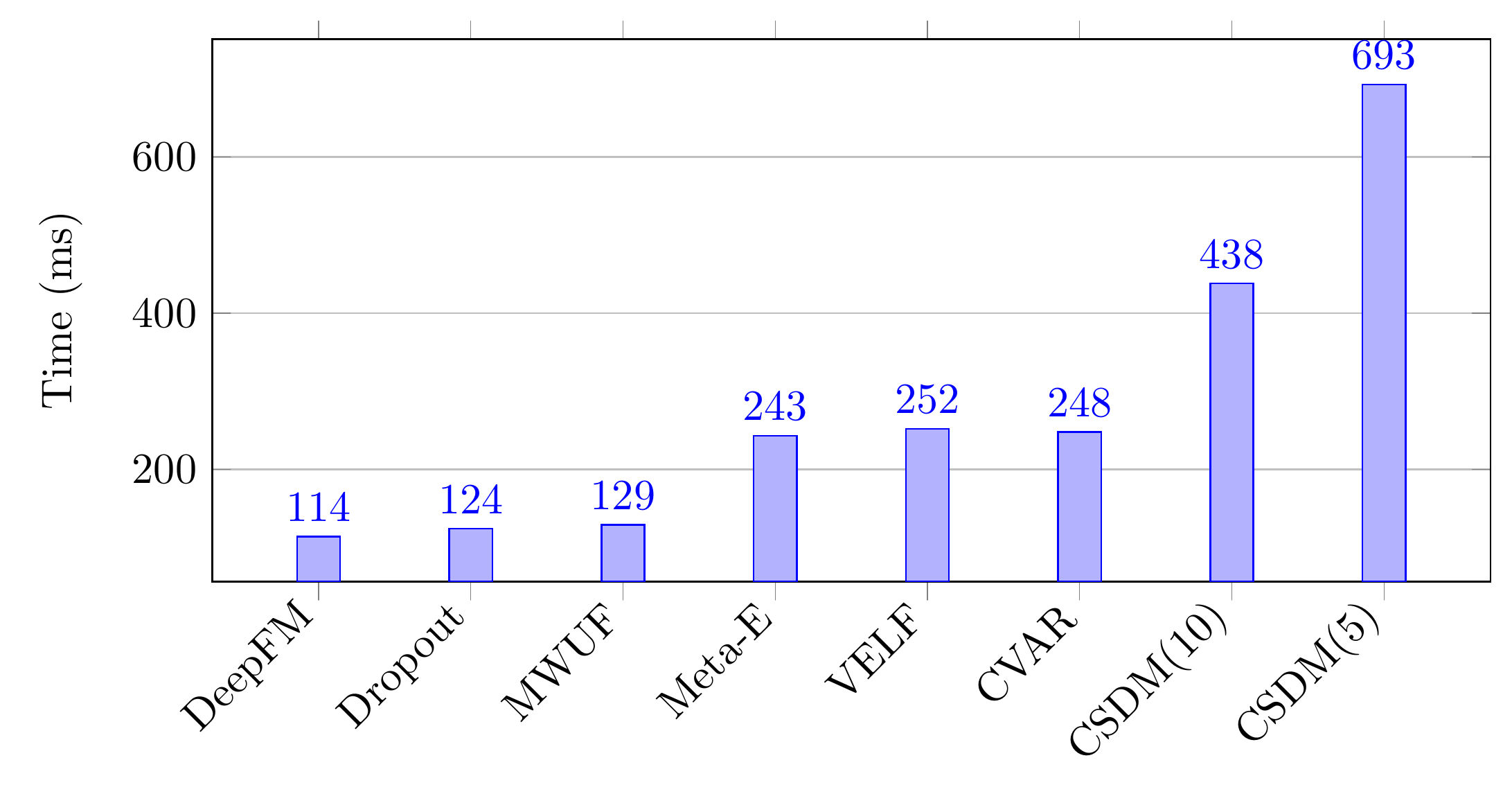}
\end{center}
\caption{The time cost for training one batch using various methods, with CSMD tested using both $s=5$ and $s=10$.}
\label{fig:overhead}
\end{figure}
We also compare the training overhead of various cold-start methods. Specifically, we measure the time cost for training one batch on the MovieLens-1M dataset and report the results in Figure \ref{fig:overhead}. The tests were conducted on a single A800 GPU with a batch size set to 2048. DeepFM is used as the backbone model.

Our method incurs higher training times due to the diffusion process. However, since our method only enhances the ID embeddings and writes the warm-up embeddings back to the original ID embeddings, there is no additional computational overhead during inference.



\section{Related Work}
\textbf{Cold-start Recommendation.}
Many approaches have been proposed to improve the recommendation of new users or items. Some of these are model-dependent, for example, Heater \cite{10.1145/3397271.3401178} and CLCRec \cite{10.1145/3474085.3475665} take CF-based models as their backbone. While some other methods are model-agnostic. 
For instance, DropoutNet \cite{10.5555/3295222.3295249} enhances the representation of users/items by applying dropout to exploit the average representations of users/items. MWUF \cite{10.1145/3404835.3462843} introduces a Meta Scaling and Shifting Network to enhance the cold ID embeddings. Meta-E \cite{10.1145/3331184.3331268} recasts the CTR prediction task as a Meta-learning \cite{Finn2017ModelAgnosticMF} problem and proposes a Meta-Embedding generator to initialize the cold ID embeddings.
VELF \cite{10.1145/3485447.3512048} and CVAR \cite{zhao2022improving} learn probabilistic embeddings to alleviate the cold-start problem in CTR prediction.
Similar to these methods, our method is also model-agnostic.
\\\textbf{Diffusion Model in Recommendation.}
Since the success of DDPM \cite{ho2020denoising} in image synthesis tasks, many researchers have attempted to leverage diffusion models in recommendation systems. DiffRec \cite{wang2023diffrec} employs diffusion models for generative recommendation. Diffusion models are also employed in \cite{Ziqiang_Cui_diff_rec} for sequential recommendation, where they are used for the semantic generation of augmented views for contrastive learning \cite{Oord2018RepresentationLW}. \cite{hanwen_du_diffrec} utilizes diffusion models to address the model collapse problems of variational auto-encoders in the sequential recommendation. \cite{10.1145/3581783.3612709} demonstrates that combining diffusion models with curriculum learning is beneficial for sequential recommendation.

Although some methods have explored the use of diffusion models for recommendations, to the best of our knowledge, employing diffusion models for cold-start in CTR prediction remains underdeveloped. This may be due to the challenges in cold-start scenarios, where we must construct a transition between ID embeddings and side information, whereas diffusion models are not directly applicable.

\section{Conclusion}
In this paper, we introduce a novel diffusion model to address the cold-start challenges in CTR prediction. It treats the embedding learning as a diffusion process. We design a non-Markovian diffusion process that enables the construction of an information flow between the side information of items and the pre-trained new item ID embeddings. Furthermore, our method can utilize both the collaborative filtering information from user action data and the side information in item features, making it applicable in both the cold-start and warm-up stages.
Experiments conducted across three distinct recommendation datasets demonstrate that the proposed method is effective in addressing the cold-start problem in CTR prediction.
\bibliography{aaai}
    \cleardoublepage
    \appendix

\section{Method}
We outline the main steps of our methodology in Algorithm \ref{algorithm}. First, a backbone model is pre-trained to provide the initial item embeddings. These pre-trained embeddings are then converted into $\mathbf{z}_0$ in the diffusion process. During the training of the diffusion model, the parameters of the backbone model are frozen. After completing the training of the diffusion model, we write the warmed-up embeddings back into the original embedding space, as shown in Figure \ref{fig:illustrate}.
\begin{algorithm}[t]
  \caption{Supervised diffusion model framework.}\label{algorithm}
  \renewcommand{\algorithmicrequire}{\textbf{Input:}}
  \renewcommand{\algorithmicensure}{\textbf{Output:}}
  \begin{algorithmic}[1]
    \REQUIRE $f_\theta$: A pretrained backbone model.
    \REQUIRE $\phi_\text{ID}^\text{new}$: Pretrained ID embeddings by the backbone.
    \REQUIRE $\mathcal{D}$: A recommendation dataset.

    \STATE Randomly initialize U-Net.
    \WHILE {not converge}
       \STATE Sample a batch sample $\mathcal{B}$ from $\mathcal{D}$.
       \STATE Extract cold ID embeddings $\mathbf{e}$ of items in $\mathcal{B}$.
       \STATE Update parameters of U-Net by optimizing $\mathcal{L}$.
    \ENDWHILE
    \WHILE {item id embedding is not replaced}
       \STATE Get the item id and related side information.
       \STATE Generate warm up embedding $\mathbf{w}$ by Equation (10). 
       \STATE Replace the item id embeddings by $\mathbf{w}$.
    \ENDWHILE
  \end{algorithmic}
\end{algorithm}
\begin{figure}[t]
\begin{center}
  \includegraphics[width=\linewidth]{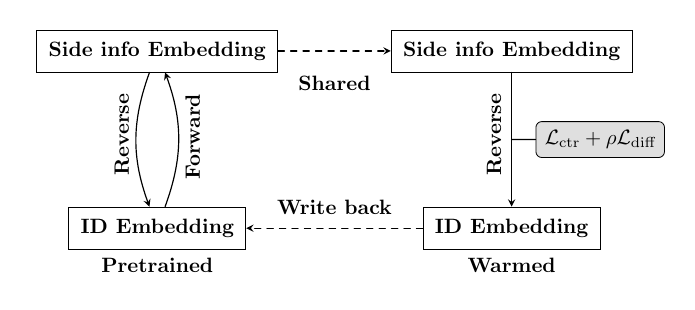}
\end{center}
\caption{An illustration of the process for generating the warm-up embeddings.}
\label{fig:illustrate}
\end{figure}
\section{Datasets}
We report the statistics of the datasets we used in our experiments in Table \ref{tb:stat}.
\\\textbf{MovieLens-1M
  \footnote{http://www.grouplens.org/datasets/movielens/}}:
It is one of the most well-known datasets for evaluating recommendation algorithms. This dataset comprises one million instances of movie ratings across thousands of movies and users. The movie features include movie ID, title, year of release, and genres, while the user features encompass age, gender, and occupation. We transfer ratings into binary (The ratings less than 4 are turned to 0 and the others are turned into 1).
\\\textbf{Taobao Display Ad Click
  \footnote{https://tianchi.aliyun.com/dataset/dataDetail?dataId=56}}: It contains 114000 randomly selected users from the Taobao website, covering 8 days of ad display and click log data (26 million records). The ad features include category ID, campaign ID, brand ID, advertiser ID, and price. User features consist of micro group ID, cms\_group\_id, gender, age, consumption level, shopping depth, occupation, and city level. The dataset includes a label of 1 for click behavior and 0 for non-click behavior.
\\\textbf{CIKM 2019 EComm AI
  \footnote{https://tianchi.aliyun.com/competition/entrance/231721}}:
It is an E-commerce dataset comprising 62 million instances. Each item is characterized by four categorical features: item ID, category ID, shop ID, and brand ID. User features encompass user ID, gender, age, and purchasing power. Each instance is tagged with a behavioral label('pv', 'buy', 'cart', 'fav'). We transform the label into a binary format (1 for a purchase action, 0 otherwise).

\section{More Experiments}
\begin{figure*}[!t]
\begin{center}
  \includegraphics[width=\linewidth]{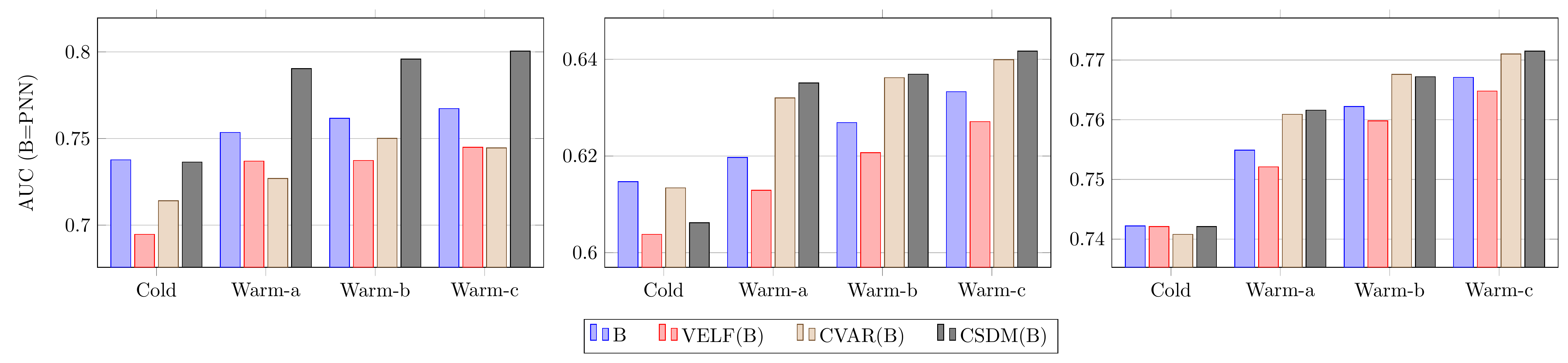}
\end{center}

\caption{AUC scores evaluated across various stages using PNN as backbone model, conducted over three datasets with 10 runs per model.}
\label{fig:pnn}
\end{figure*}
\textbf{Generalization Experiments.}
Beside applying DeepFM, Wide \& Deep, and DCN as backbone models, we also conduct experiments on PNN. PNN \cite{PNN7837964} is a CTR prediction model that employs a product layer to explore the interactions among inter-field categories.
We present the results in Figure \ref{fig:pnn}. The results show that our method outperforms others in most cases.
\begin{table}[t]
  \begin{center}
    \begin{tabular}{l|r|r|r}
      \toprule
      Dataset & MovieLens-1M & Taobao AD & CIKM 2019 \\
      \midrule      
      \#user   & 6040         & 1141729   &1050000    \\
      \#item   & 3706         & 864811    &3934201    \\
      \#instance & 1000209    & 25029435  & 62428486  \\
      \bottomrule
    \end{tabular}
  \end{center}
  \caption{Statistics of datasets used in our experiments.}
  \label{tb:stat}
\end{table}

\begin{table}[t]
  \setlength{\tabcolsep}{3pt}  
  \begin{center}
    \begin{tabular}{c|c| c|c|c|c}
      \toprule
      Dropout & Dataset & Cold & Warm-a & Warm-b & Warm-c \\
     \midrule      
     w  & ML-1M  & \textbf{0.7362} & 0.7971 & 0.8020 & \textbf{0.8059} \\
     w/o  & ML-1M & 0.7360 & 0.7971 & 0.8020 & 0.8057 \\
     \midrule
     w  & TaobaoAD & 0.6097 & 0.6352 & \textbf{0.6387} & \textbf{0.6436} \\
     w/o& TaobaoAD & \textbf{0.6103} & \textbf{0.6361} & 0.6382 & 0.6428 \\
     \midrule
     w  & CIKM     & 0.7425 & \textbf{0.7601} & \textbf{0.7657} & \textbf{0.7706} \\
     w/o& CIKM     & 0.7425 & 0.7541 & 0.7629 & 0.7681 \\
    \bottomrule
    \end{tabular}
  \end{center}
  \caption{An ablation test on the dropout function of diffusion models: "w" indicates that dropout is enabled, whereas "w/o" signifies that dropout is disabled. ML-1M stands for MovieLens-1M. DCN is used as the backbone model.}
  \label{tb:drop-dcn}
\end{table}

\textbf{Ablation Study.}
We also conduct an ablation analysis on the diffusion model concerning the dropout mechanism across three datasets using DCN as the backbone model. The results are shown in Table \ref{tb:drop-dcn}.
We can observe similar results:
on the MovieLens-1M dataset, the inclusion of dropout provides minimal improvement. In contrast, on the CIKM 2019 dataset, incorporating dropout into the diffusion process positively impacts and enhances the model's performance.

\section{Proofs}
\subsection{Definition of $q_\sigma(\mathbf{z}_{1:T}|\mathbf{z_0}, \mathbf{h})$}
We rewrite the definition of $q_\sigma(\mathbf{z}_{1:T}|\mathbf{z_0}, \mathbf{h})$ in the paper:
\begin{align}
  q_\sigma(\mathbf{z}_{1:T}|\mathbf{z}_0, \mathbf{h}) :=
  q_\sigma(\mathbf{z}_{T} | \mathbf{z}_0, \mathbf{h})
  \prod_{t=2}^{T} q_\sigma\left(\mathbf{z}_{t-1}|\mathbf{z}_t, \mathbf{h}, \mathbf{z}_0\right)
  \label{eq:def-dist}
\end{align}

\subsection{Definition of $q_\sigma(\mathbf{z}_{t-1}|\mathbf{z}_t,\mathbf{z}_0, \mathbf{h})$}
As shown in our main paper, we have defined: 
\begin{align}
  q_\sigma(\mathbf{z}_{t-1}|\mathbf{z}_t,\mathbf{z}_0, \mathbf{h}) =
  \mathcal{N}(\mathbf{z}_{t-1}|\kappa_t \mathbf{z}_t + \lambda_t \mathbf{z}_0 + \nu_t \mathbf{h}, \sigma_t^2 \mathbf{I})
  \label{eq:def-post-mean}
\end{align}
where, 
\begin{align}
  &\kappa_t = \sqrt{\frac{1-\alpha_{t-1} - \sigma_t^2}{1-\alpha_t}} \notag \\
  &\lambda_t = \sqrt{\alpha_{t-1}} - \sqrt{\alpha_t} \sqrt{\frac{1-\alpha_{t-1} - \sigma_t^2}{1-\alpha_t}} \notag \\
  &\nu_t = \sqrt{c_{t-1}} - \sqrt{c_t} \sqrt{\frac{1-\alpha_{t-1} - \sigma_t^2}{1-\alpha_t}} \notag
\end{align}

\begin{lemma}
Given the definitions of \( q_\sigma(\mathbf{z}_{1:T} | \mathbf{z}_0, \mathbf{h}) \) in Equation (\ref{eq:def-dist}) and \( q_\sigma(\mathbf{z}_{t-1} | \mathbf{z}_t, \mathbf{z}_0, \mathbf{h}) \) in Equation (\ref{eq:def-post-mean}), we have:
  \begin{equation}
      q_\sigma(\mathbf{z}_{t}| \mathbf{z}_0, \mathbf{h}) =
  \mathcal{N}(\sqrt{\alpha_{t}}\mathbf{z}_0 + \sqrt{c_{t}}\mathbf{h}, (1-\alpha_{t})\mathbf{I})
  \label{eq:forward-step}
  \end{equation}
\end{lemma}
\begin{proof}
Following \cite{song2020denoising}, we prove the statement by induction. 
First, for $t = T$, we already have:
  \begin{equation}
      q_\sigma(\mathbf{z}_{t}| \mathbf{z}_0, \mathbf{h}) =
  \mathcal{N}(\sqrt{\alpha_{t}}\mathbf{z}_0 + \sqrt{c_{t}}\mathbf{h}, (1-\alpha_{t})\mathbf{I})
  \end{equation}
  To prove the Equation (\ref{eq:forward-step}) holds for $t < T$, we have
\begin{equation}
  q_\sigma(\mathbf{z}_{t-1} | \mathbf{z}_0, \mathbf{h}) :=
  \int_{\mathbf{z}_t}
  q_\sigma(\mathbf{z}_t | \mathbf{z}_0, \mathbf{h})
  q_\sigma(\mathbf{z}_{t-1} | \mathbf{z}_t, \mathbf{z}_0, \mathbf{h})    
\end{equation}
Since \( q_\sigma(\mathbf{z}_{t} | \mathbf{z}_0, \mathbf{h}) \) and \( q_\sigma(\mathbf{z}_{t-1} | \mathbf{z}_t, \mathbf{z}_0, \mathbf{h}) \) are both Gaussian, from \cite{10.5555/1162264} (Equation 2.115), we know that \( q_\sigma(\mathbf{z}_{t-1} | \mathbf{z}_0, \mathbf{h}) \) is also Gaussian. We denote it as \( \mathcal{N}(\mu_{t-1}, \Sigma_{t-1}) \).
where
\begin{align}
  \mu_{t-1} &= \kappa_t (\sqrt{\alpha_t} \mathbf{z}_0 + \sqrt{c_t}\mathbf{h}) + \lambda_t \mathbf{z}_0 + \nu_t \mathbf{h} \notag \\
  &=\sqrt{\alpha_{t-1}} \mathbf{z}_0  + \sqrt{c_{t-1}}\mathbf{h}
\end{align}

\begin{align}
  \Sigma_{t-1} &= \sigma_t^2 \mathbf{I} + \frac{1-\alpha_{t-1} - \sigma_t^2}{1-\alpha_t} (1-\alpha_{t})\mathbf{I} \\
  &= (1-\alpha_{t-1})\mathbf{I} 
\end{align}
Therefore,
\begin{equation}
      q_\sigma(\mathbf{z}_{t-1}| \mathbf{z}_0, \mathbf{h}) =
      \mathcal{N}(\sqrt{\alpha_{t-1}}\mathbf{z}_0 + \sqrt{c_{t-1}}\mathbf{h}, (1-\alpha_{t-1})\mathbf{I})
\end{equation}
By applying induction, we establish that Equation (\ref{eq:forward-step}) holds for \( t \leq T \).
\end{proof}

\end{document}